\documentclass[11pt, oneside]{article}   	
\usepackage{geometry}                		
\usepackage[parfill]{parskip}    		
\usepackage{graphicx}				
\usepackage{amssymb, amsmath, amsthm}

\newtheorem{theorem}{Theorem}
\newtheorem{lemma}{Lemma}
\newtheorem{proposition}{Proposition}
\newtheorem{remark}{Remark}
\newtheorem{example}{Example}
\newtheorem{corollary}{Corollary}

\title{Sum-of-Square Proof for Brascamp-Lieb Type Inequalities}

\date{}							

\author{Zhixian Lei, Yueqi Sheng}

\begin{document}
\maketitle

\begin{abstract}
\emph{Brascamp-Lieb} inequalities \cite{brascamp1976best} is an important mathematical tool in analysis, geometry and information theory. There are various ways to prove Brascamp-Lieb inequality such as heat flow method \cite{bennett2008brascamp}, Brownian motion \cite{lehec2013short} and subadditivity of the entropy \cite{carlen2009subadditivity}. While Brascamp-Lieb inequality is originally stated in Euclidean Space, \cite{christ2013optimal} discussed Brascamp-Lieb inequality for discrete Abelian group and \cite{barthe2010correlation} discussed Brascamp-Lieb inequality for Markov semigroups.

Many mathematical inequalities can be formulated as algebraic inequalities which asserts some given polynomial is nonnegative. In 1927, Artin proved that any nonnegative polynomial can be represented as a sum of squares of rational functions \cite{delzell1984continuous}, which can be further formulated as a polynomial certificate of the nonnegativity of the polynomial. This is a \emph{Sum-of-Square proof} of the inequality. 
The Sum-of-Square proof can be captured by 
\emph{Sum-of-Square algorithm} which is a powerful tool for optimization and computer aided proof. For more about Sum-of-Square algorithm, see \cite{barak2014sum}. 

In this paper, we give a Sum-of-Square proof for some special settings of Brascamp-Lieb inequality following \cite{christ2015holder} and \cite{bennett2008brascamp} and discuss some applications of Brascamp-Lieb inequality on Abelian group and Euclidean Sphere. If the original description of the inequality has constant degree and $d$ is constant, the degree of the proof is also constant. Therefore, low degree sum of square algorithm can fully capture the power of low degree finite Brascamp-Lieb inequality. 
\end{abstract}

\section{Introduction}
\subsection{Brascamp-Lieb inequality}
Many important inequalities including Holder's inequality, Loomis-Whitney inequality, Young's convolution inequality, hypercontractivity inequalities are special case of Brascamp-Lieb inequality introduced by \cite{brascamp1976best}. 
The original form of Brascamp-Lieb inequality on Euclidean Space $\mathbb{R}^n$ is 
\begin{equation}
\label{BL inequality Euclidean}
\int_{x \in \mathbb{R}^n}\prod_{j=1}^m(f_j(B_jx))^{p_j}dx \leq C \prod_{j=1}^m \left(\int_{x_j \in \mathbb{R}^{n_j}}f_j(x_j)dx_j\right)^{p_j}
\end{equation}
where 
\begin{enumerate}
\item $B_j : \mathbb{R}^n \to \mathbb{R}^{n_j}$ are linear surjective maps
\item $f_j: \mathbb{R}^{n_j} \to \mathbb{R}$ are nonnegative functions 
\item $p_j$ are nonnegative reals. 
\item $C$ is positive and independent of $f_j$
\end{enumerate}
Following theorem \cite{bennett2008brascamp} gives the condition when (\ref{BL inequality Euclidean}) holds. 
\begin{proposition} (\ref{BL inequality Euclidean}) holds if and only if 
\begin{enumerate}
\item $n = \sum_j p_j n_j$
\item dim$(V)$ $\leq \sum_jp_j$dim$(B_jV)$ for all subspaces $V$ of $\mathbb{R}^n$
\end{enumerate}
\end{proposition}
The inequality is saturated when $f_j$ are centered Gaussian functions and the optimal $C$ has the form:
\begin{equation}
C = \left[\sup\frac{\prod_j(\det X_j)^{p_j}}{\det\left(\sum_j p_jB_j^TX_jB_j\right)}\right]^{1/2} 
\end{equation}
where the supreme is taken over all positive semidefinite matrix $X_j$ in dimension $n_j$. Moreover, the value of optimal $C$ can be $(1 + \epsilon)$-approximated in time $\textrm{poly}(\frac{1}{\epsilon})$(see \cite{garg2016algorithmic})


To formulate (\ref{BL inequality Euclidean}) in an algebraic form, we replace integration by taking finite summation.
\begin{equation}
\label{BL inequality sum}
\sum_{x \in V}\prod_{j=1}^m(f_j(B_jx)) \leq C \prod_{j=1}^m \left(\sum_{x_j \in B_jV} f_j(x_j)^{1/p_j}\right)^{p_j}
\end{equation}
Where $V$ is a finite space. (\ref{BL inequality sum}) can also be simplified as
\begin{equation}
\label{BL sum}
\sum_V\prod_{j=1}^m(f_j \circ B_j) \leq C \prod_{j=1}^m \|f_j\|_{1/p_j}
\end{equation}
where $\| f_j \|_{1/p_j}= (\sum_{x_j \in B_jV} f_j(x_j)^{1/p_j})^{p_j}$.

\subsection{Sum-of-Square proof}
A Sum-of-Square proof for polynomial $P \geq 0$ is to give the following certificate
\begin{equation}
\label{certificate}
S_1P - S_2 = 0
\end{equation}
where $S_1$ and $S_2$ are sum of square of polynomials. The degree of the proof is the degree of (\ref{certificate}). 

For the simplicity of exposition, we will give a Sum-of-Square proof in an iterative way with following deduction rules. 
\begin{eqnarray*}
P_1 \geq 0, P_2 \geq 0 & \Longrightarrow  & P_1 + P_2 \geq 0 \\
P_1 \geq 0, P_2 \geq 0 & \Longrightarrow &  P_1P_2 \geq 0 \\
& \Longrightarrow  & P_1^2 \geq 0
\end{eqnarray*}
where $P_1, P_2$ are polynomials. To prove $P \geq 0$, in the end we should derive 
\begin{equation*}
SP \geq 0, S \geq 0
\end{equation*}
for some polynomial $S$. The degree of the proof is also accumulated with deduction. 
\begin{eqnarray*}
\deg(P_1 + P_2) & = &\max\{\deg(P_1), \deg(P_2)\} \\
\deg(P_1P_2) & =  &\deg(P_1) + \deg(P_2) \\
\deg(P_1^2) & = & 2\deg(P_1)
\end{eqnarray*}
The degree of the proof is the largest degree which appears in the deduction. 

In sum of square algorithm, \emph{Pseudo distribution} is a dual certificate for Sum-of-Square proof. Pseudo distribution is not necessary a real distribution. Instead the only requirements for a degree d pseudo distribution is to have a corresponding \emph{pseudo expectation} $\tilde{\mathbb{E}}$ to satisfy
\begin{enumerate}
\item $\tilde{\mathbb{E}}1 = 1$
\item $\tilde{\mathbb{E}}P + \tilde{\mathbb{E}}Q = \tilde{\mathbb{E}}(P + Q)$ for all polynomial $P$ and $Q$ of degree no more than $d$
\item $\tilde{\mathbb{E}}P^2 \geq 0$ for all polynomial $P$ of degree no more than $d/2$
\end{enumerate} 
If degree $d$ Sum-of-Square cannot prove  $P \geq 0$, then there exists a degree $d$ pseudo distribution satisfies $\tilde{\mathbb{E}}P < 0$. In this way, degree $d$ pseudo distribution captures the power the degree $d$ Sum-of-Square proof. Pseudo distribution is a more general notion than Sum-of-Square proof. We can implicitly evaluate pseudo expectation $\tilde{\mathbb{E}}f$ for any function $f$ in any space without giving a polynomial form. 

\subsection{Our result}
Consider $V$ as a finite subset of $\mathbb{Z}^n$ with a set of linear projections  $\{B_j : \mathbb{Z}^n \to \mathbb{Z}^{n_j}\}$. Let $\{f_j: V_j \to \mathbb{R}\}$ be a set of non-negative functions and define $V_j = B_j(V)$. Then we have
\begin{theorem}
If all $p_j$ satisfies $p_j \geq 0$ and 
\begin{equation*}
\dim(W) \leq \sum_jp_j\dim(B_jW)\textrm{ for all subspaces }W\textrm{ of }V
\end{equation*}
Then 
\begin{equation}
\label{sos proof}
\sum_V\prod_{j=1}^m(f_j \circ B_j) \leq \prod_{j=1}^m \|f_j\|_{1/p_j}
\end{equation}
can be proved by by degree $O(n^mm^{m/2} + s\sum_{j=1}^ms_j)$ Sum-of-Square where $s_j, s$ are integers, $s_j/s = p_j$, $s$ is the least common denominator of all $p_j$.
\end{theorem}
Note that the degree of (\ref{sos proof}) are $s\sum_{j=1}^ms_j$, if we take the degree of original expression of the inequality to be constant, $m, s, s_j = O(1)$, then the degree of the pseudo expectation is also a constant, and the degree of Sum-of-Square proof becomes poly$(n^{O(1)})$. 
\begin{remark}
In fact, consider space $\mathbb{Z}^n$ is quite general because we can reduce $\mathbb{Q}^n$ to $\mathbb{Z}^n$ by normalizing every points and projections to integral ones since the inequality only involves finite many points and projections. And in fact we can even generalize this discrete inequality on any set of points if we can embed these points on $Z^n$ with proper definition of linear projection. 
\end{remark}
We will see that for this finite discrete Brascamp-Lieb inequality, there are still many famous inequalities can be formulated in this way.
\begin{example}[Holder's inequality]
When $n = 1$ and $m = 2$, consider non-negative functions $f$ and $g$ with all projections to be identity we have 
\begin{equation*}
\sum f(x)g(x) \leq \lVert f\rVert_{1/p} \lVert g \rVert_{1/q}
\end{equation*}
when $p + q = 1$. When $p = q = 1/2$ this gives Cauchy-Schwarz ineqaulity
\end{example}
\begin{example}[Loomis-Whitney inequality]
when $m = n$, $B_j$ are projections to the orthogonal complement to each coordinate and all $p_j$ are $1/(n-1)$ the Brascamp-Lieb inequality gives exactly the Loomis-Whitney inequality. For instance, when $n = 3$ we have
\begin{equation*}
\sum_{x, y, z}f(y, z)g(x, z)h(x, y) \leq \lVert f \rVert_2 \lVert g \rVert_2 \lVert h \rVert_2
\end{equation*}
which has deep interpretations in geometry.
\end{example}

\section{Sum-of-Square proof for Holder's inequality}
In this section, we give the Sum-of-Square proof of \emph{Holder's inequality} and analyze the degree of the proof for future use. First we give the proof for \emph{Cauchy-Schwarz inequality}. 
\begin{lemma} [Cauchy-Schwarz inequality]
\begin{equation*}
\tilde{\mathbb{E}}fg \leq (\tilde{\mathbb{E}}f^2)^{1/2}(\tilde{\mathbb{E}}g^2)^{1/2}
\end{equation*}
is satisfied by degree $2(\deg(f)+\deg(g))$ pseudo distribution
\end{lemma}
\begin{proof}
By $(f - g)^2 \geq 0$ we have 
\begin{equation*}
\tilde{\mathbb{E}}fg \leq \frac{1}{2}\tilde{\mathbb{E}}f^2 + \frac{1}{2}\tilde{\mathbb{E}}g^2
\end{equation*}
Let $f' = f/(\tilde{\mathbb{E}}f^2)^{1/2}$ and $g' = g/(\tilde{\mathbb{E}}g^2)^{1/2}$ then 
\begin{equation*}
\tilde{\mathbb{E}}f'g' = \frac{\tilde{\mathbb{E}}fg}{(\tilde{\mathbb{E}}f^2)^{1/2}(\tilde{\mathbb{E}}g^2)^{1/2}} \leq \frac{1}{2}\tilde{\mathbb{E}}f'^2 + \frac{1}{2}\tilde{\mathbb{E}}g'^2 = 1
\end{equation*}
therefore
\begin{equation*}
\tilde{\mathbb{E}}fg \leq (\tilde{\mathbb{E}}f^2)^{1/2}(\tilde{\mathbb{E}}g^2)^{1/2}
\end{equation*}
\end{proof}
By taking the pseudo distribution as uniform distribution, we also get the sum of square proof of Cauchy-Schwarz inequality
\begin{corollary} [Cauchy-Schwarz inequality]
\begin{equation*}
\sum f(x)g(x) \leq \|f\|_{2}\|g\|_{2}
\end{equation*}
has a degree $2(\deg(f)+\deg(g))$ Sum-of-Square proof
\end{corollary}
Using Cauchy-Schwarz inequality, we can further prove Holder's inequality
\begin{lemma} [Holder's inequality] when $p + q = 1$
\begin{equation*}
\tilde{\mathbb{E}}fg \leq (\tilde{\mathbb{E}}f^{1/p})^p(\tilde{\mathbb{E}}g^{1/q})^q
\end{equation*}
is satisfied by degree $s(s_1 + s_2)(\deg(f) + \deg(g))$ pseudo distribution where $s, s_1, s_2$ are integers, $p = s_1/s$, $q = s_2/s$, $s$ is the least common denominator of $p$ and $q$
\end{lemma}
\begin{proof}
We can iteratively approximate the inequality using Cauchy-Schwarz inequality. Since $p + q = 1$, one of $p, q$ is no less than $1/2$. Without loss of generality, assume $q \geq 1/2$. If $q = 1/2$, the inequality becomes Cauchy-Schwarz inequality. If $q > 1/2$, We have by Cauchy-Schwarz inequality
\begin{equation*}
\tilde{\mathbb{E}}fg = \tilde{\mathbb{E}}fg^{1-1/2q}g^{1/2q} \leq (\tilde{\mathbb{E}}f^2g^{2-1/q})^{1/2}(\tilde{\mathbb{E}}g^{1/q})^{1/2}
\end{equation*}
It remains to prove $(\tilde{\mathbb{E}}f^2g^{2-1/q})^{1/2} \leq (\tilde{\mathbb{E}}f^{1/p})^{p}(\tilde{\mathbb{E}}g^{1/q})^{q-1/2}$. Notice that the exponent $p, q-1/2$ on right hand side is decreased. In next iteration, we will subtract the max of $p$ and $q-1/2$ by 1/4. In this way, we can iteratively approximate Holder's inequality. The degree is Sum-of-Square proof is determined by the fractional expression of $p$ and $q$. If we assume the degree in the expression of original inequality to be constant, The degree of Sum-of-Square proof for Holder's inequality is also constant. 
\end{proof}
\begin{example}
\begin{equation*}
\tilde{\mathbb{E}}fg \leq (\tilde{\mathbb{E}}f^{8/3})^{3/8}(\tilde{\mathbb{E}}g^{8/5})^{5/8}
\end{equation*}
is satisfied by constant degree pesudo distribution. 
\end{example}
\begin{proof}
\begin{eqnarray*}
& & \tilde{\mathbb{E}}fg \\
& \leq & (\tilde{\mathbb{E}}f^2g^{2/5})^{1/2} (\tilde{\mathbb{E}}g^{8/5})^{1/2} \\
& \leq & (\tilde{\mathbb{E}}f^{8/3})^{1/4} (\tilde{\mathbb{E}}f^{4/3}g^{4/5})^{1/4}  (\tilde{\mathbb{E}}g^{8/5})^{1/2} \\
& \leq & (\tilde{\mathbb{E}}f^{8/3})^{1/4} (\tilde{\mathbb{E}}f^{8/3})^{1/8} (\tilde{\mathbb{E}}g^{8/5})^{1/8}  (\tilde{\mathbb{E}}g^{8/5})^{1/2} \\
& = & (\tilde{\mathbb{E}}f^{8/3})^{3/8}(\tilde{\mathbb{E}}g^{8/5})^{5/8}
\end{eqnarray*}
\end{proof}
Also, by assuming pseudo distribution as uniform distribution, we have
\begin{corollary} [Holder's inequality] when $p + q = 1$
\begin{equation*}
\sum fg \leq \|f\|_{1/p}\|g\|_{1/q}
\end{equation*}
is satisfied by degree $s(s_1 + s_2)(\deg(f) + \deg(g))$ pseudo distribution where $s, s_1, s_2$ are integers, $p = s_1/s$, $q = s_2/s$, $s$ is the least common denominator of $p$ and $q$
\end{corollary}
From next section we will use Holder's inequality to prove more general Brascamp-Lieb inequality without considering the degree increased by Holder's inequality since the degree increasing is explicitly shown in the expression. 

\section{Reduce Brascamp-Lieb inequality to extreme points}
Recall that for finite $V \subseteq \mathbb{R}^n$ and projections $B_j : V \to V_j$, 
we will give a Sum-of-Square proof of
\begin{equation}
\label{sos}
\sum_V\prod_{j=1}^m(f_j \circ B_j) \leq \prod_{j=1}^m \|f_j\|_{1/p_j}
\end{equation}
Let $p = (p_1, p_2, \ldots, p_m)$ Define $P(V)$ as the feasible region of $p$ in (\ref{sos}) 
\begin{equation*}
P(V) = \{p \mid \textrm{for all } j, p_j \geq 0, \sum_{j=1}^m p_j = 1\} 
\end{equation*}
$P(V)$ is a bounded polytope. For the feasible region of $p$ in (\ref{sos}), notice $p_j$ is not upper bounded. But in fact, we can require $p_j \leq 1$ for all $j$. Define $Q(V)$ as the feasible region of $p$ in (\ref{sos})
\begin{equation*}
Q(V) = \{p \mid p \in [0, 1]^m, \dim(W) \geq  \sum_jp_j\dim(B_jW)\textrm{ for all subspace } W\textrm{ of } V\}
\end{equation*}
Next we prove the validity of requiring $p_j \leq 1$
\begin{lemma}
If (\ref{sos}) has Sum-of-Square proof for all $p \in Q(V)$, then (\ref{sos}) has Sum-of-Square proof for all feasible $p$.
\end{lemma}
\begin{proof}
We want to prove (\ref{sos}) for feasible $p$ with some $p_j > 1$, let $p'_j = p_j$ when $p_j \leq 1$ and $p'_j = 1$ when $p_j > 1$, then $p' \in Q(V)$ so we have
\begin{equation*}
\sum_V \prod_{j=1}^mf_j \circ B_j \leq \prod_{j=1}^m \| f_j\|_{1/p_j'}
\end{equation*}
Further we can prove that $\|f_j\|_1 \leq \|f_j\|_{1/p_j}$ for $p_j > 1$. Let $f_j' = f_j/\|f_j\|_1$ we have
\begin{equation*}
\|f_j'\|_{1/p_j} = \left\| \frac{f_j}{\|f_j\|_1} \right\|_{1/p_j} \geq 1 = \|f_j'\|_1
\end{equation*}
Combining above gives Sum-of-Square proof for $p$
\begin{equation*}
\sum_V \prod_{j=1}^mf_j \circ B_j \leq \prod_{j=1}^m \| f_j\|_{1/p_j}
\end{equation*}
\end{proof}
$Q(V)$ is a also bounded polytope. Following lemma shows that we can prove (\ref{sos}) for $p \in P(V)$ and $p \in Q(V)$ respectively if we can prove (\ref{sos}) for extreme points of $P(V)$ and $Q(V)$. 
\begin{lemma}
Suppose (\ref{sos}) holds for $p_1, p_2$, then (\ref{sos}) holds for $p = \theta p_1 + (1 - \theta)p_2$ for all $\theta \in [0, 1]$
\end{lemma}
\begin{proof}
Suppose we have (\ref{sos}) for $p_1$ and $p_2$
\begin{eqnarray*}
\sum\prod_{j=1}^mf_j\circ B_j & \leq & \prod_{j=1}^m(\sum f_j^{1/p_1})^{p_1} \\
\sum\prod_{j=1}^mf_j\circ B_j & \leq & \prod_{j=1}^m(\sum f_j^{1/p_2})^{p_2}
\end{eqnarray*}
Replace $f_j$ by $f_j^{p_1/p}$ and $f_j := f_j^{p_2/p}$ respectively we get
\begin{eqnarray*}
\sum\prod_{j=1}^mf_j^{p_1/p} \circ B_j & \leq & \prod_{j=1}^m(\sum f_j^{1/p})^{p_1} \\
\sum\prod_{j=1}^mf_j^{p_2/p} \circ B_j & \leq & \prod_{j=1}^m(\sum f_j^{1/p})^{p_2}
\end{eqnarray*}
Multiply above inequality with exponents $\theta$ and $(1 - \theta)$
\begin{equation*}
\left(\sum\prod_{j=1}^mf_j^{p_1/p} \circ B_j \right)^{\theta} \left(\sum\prod_{j=1}^mf_j^{p_2/p} \circ B_j \right)^{1 - \theta} \leq \prod_{j=1}^m(\sum f_j^{1/p})^{p}
\end{equation*}
By Holder's inequality
\begin{equation*}
\sum\prod_{j=1}^mf_j \circ B_j \leq  \left(\sum\prod_{j=1}^mf_j^{p_1/p} \circ B_j \right)^{\theta} \left(\sum\prod_{j=1}^mf_j^{p_2/p} \circ B_j \right)^{1 - \theta}
\end{equation*}
Finally we have 
\begin{equation*}
\sum\prod_{j=1}^mf_j \circ B_j \leq \prod_{j=1}^m(\sum f_j^{1/p})^{p}
\end{equation*}
\end{proof}
By replacing pseudo expectation with summation we also have
\begin{corollary}
Suppose (\ref{sos}) holds for $p_1, p_2$, then (\ref{sos}) holds for $p = \theta p_1 + (1 - \theta)p_2$ for all $\theta \in [0, 1]$
\end{corollary}
Next section we give the proof of (\ref{sos}) on extreme points of $P(V)$ and $Q(V)$ respectively.

\section{Prove Brascamp-Lieb inequality on extreme points}
First we prove (\ref{sos}) for extreme points of $P(V)$. The extreme points $p$ in $P(V)$ have following form: there is one $j$ such that $p_j = 1$, and for all other $j' \neq j$, $p_{j'} = 0$. In this case (\ref{sos}) becomes
\begin{equation*}
\sum\prod_{j=1}^m(f_j \circ B_j) \leq \sum f_j\prod_{j' \neq j}^m (\sum \max f_{j'})
\end{equation*}
This holds trivially. So we complete the prove of (\ref{sos}). For the degree of the proof, notice that the highest degree appearing in the proof is in the last expression of the inequality. So the degree of the pseudo distribution is $O(s\sum_js_j)$.

Then we give the Sum-of-Square proof of (\ref{sos}) on extreme points $p \in Q(V)$ by induction on dimension $\dim(V)$. When $\dim(V) = 0$, both left hand side and right hand side become $\prod_{j=1}^mf_j(0)$, (\ref{sos}) holds trivially. When $\dim(V) > 0$, suppose (\ref{sos}) holds for any space with dimension less than $\dim(G)$. We will give a Sum-of-Square proof for (\ref{sos}) on $V$.  

Consider a nontrivial subspace $W$ of $V$ , we can decompose $V = W \oplus V/W$ and decompose $B_j$ into $B^W_j$ and $B^{V/W}_j$ accordingly
\begin{enumerate}
\item $B_j^W : W \to V_j$, restriction of $B_j$ to $W$
\item $B_j^{V/W} : V/W \to B_jV/B_jW$, $B_j^{V/W}(x + W) = B_j(x) + B_jW$
\end{enumerate}
By $B^W$ and $B^{V/W}$, we can define the feasible space $Q(W)$ and $Q(V/W)$ for $W$ and $V/W$ accordingly. 
\begin{enumerate}
\item $Q(W) = \{p \mid p \in [0, 1]^m, \dim(W') \geq  \sum_jp_j\dim(B_jW')\textrm{ for all subspace } W'\textrm{ of } W\}$
\item $Q(V/W) = \{p \mid p \in [0, 1]^m, \dim(W') \geq  \sum_jp_j\dim(B_jW')\textrm{ for all subspace } W'\textrm{ of } V/W\}$
\end{enumerate}
Following proposition \cite{bennett2008brascamp} gives the condition for $p$ to be feasible in $W$ and $V/W$.
\begin{lemma}
Let $W$ be a subspace of $V$, if $\dim(W) =  \sum_jp_j\dim(B_jW)$, 
\begin{equation*}
p \in Q(V) \iff p \in Q(W) \cap Q(V/W)
\end{equation*}
\end{lemma}
We can also define functions $f^W_j$ and $f^{V/W}_j$ on $W$ and $V/W$ as follows
\begin{enumerate}
\item $f^W_j: B_jW \to \mathbb{R}$, restriction of $f_j$ to $B_jW$
\item $f^{V/W}_j: B_jV / B_jW \to \mathbb{R}$, $f^{V/W}(x + B_jW) = \left(\sum_{y \in B_jW} f(x + y)^{1/p_j}\right)^{p_j}$
\end{enumerate}
With these definitions, we can reduce (\ref{sos}) into lower dimension cases for some $W$. 
\begin{lemma}
Suppose there exists nontrivial subspace $W$ such that $\dim(W) =  \sum_jp_j\dim(B_jW)$, by induction hypothesis, we can prove (\ref{sos}) for $p$ by Sum-of-Square.
\end{lemma}
\begin{proof}
(\ref{sos}) can be written as
\begin{equation*}
\sum_{y \in V/W}\sum_{x \in W}\prod_{j=1}^mf_j(B_jx + B_jy) \leq \prod_{j=1}^m \left(\sum_{y \in B_jV/B_jW}\sum_{x \in B_jW} f(x + y)^{1/p_j}\right)^{p_j}
\end{equation*}
Let $f_j': B_jW \to \mathbb{R}$ such that $f_j'(B_jx) = f_j(B_jx + B_jy)$ for fixed $y$, apply the induction hypothesis on $f_j'$ in $W$
\begin{equation*}
\label{ih}
\sum_{x \in W}\prod_{j=1}^mf_j(B_jx + B_jy) = \sum_{x \in W}\prod_{j=1}^mf_j'(B_jx) \leq \prod_{j=1}^m \| f'^W_j \|_{1/p_j} = \prod_{j=1}^m f^{V/W}_j(B_jy)
\end{equation*}
Apply induction hypothesis on $f^{V/W}_j$ in $V/W$
\begin{equation*}
\sum_{y \in V/W} \prod_{j=1}^m f^{V/W}_j(B_jy) \leq \prod_{j=1}^m \|f^{V/W}_j \|_{1/p_j}
\end{equation*}
Combining above gives
\begin{equation*}
\sum_{y \in V/W}\sum_{x \in W}\prod_{j=1}^mf_j(B_jx + B_jy) \leq \sum_{y \in V/W} \prod_{j=1}^m f^{V/W}_j(B_jy) \leq \prod_{j=1}^m \|f_j\|_{1/p_j}
\end{equation*}
\end{proof}
Now we should consider the case when there is no nontrivial subspace $W$ satisfying $\dim(W) =  \sum_jp_j\dim(B_jW)$. Following proposition \cite{christ2015holder} characterizes this case.
\begin{proposition}
If there is no nontrivial subspace $W$ satisfying $\dim(W) =  \sum_jp_j\dim(B_jW)$, then $p \in \{0, 1\}^m$
\end{proposition}
Next we give a Sum-of-Square proof for all feasible $p \in \{0, 1\}^m$
\begin{lemma}
For all feasible $p \in \{0, 1\}^m$, we have a Sum-of-Square proof for (\ref{sos}).
\end{lemma}
\begin{proof}
Since $p$ is feasible, apply condition $p \in Q(V)$ on subspace $\bigcap_{p_j = 1}\ker(B_j)$
\begin{equation*}
\dim\left(\bigcap_{p_j = 1}\ker(B_j)\right) = 0
\end{equation*}
Then there is no $x \neq y \in V$ such that $B_jx  = B_jy$ for all $j$. If we only consider $f_j$ such that $p_j = 1$, all terms on the left hand side also appears on the right hand side at least once, so
\begin{equation*}
\sum_{x \in V}\prod_{p_j=1}f_j(B_jx) \leq \prod_{p_j=1} \| f_j \|_1
\end{equation*}
when $p_j = 0$,  $\|f_j\|_{1/p_j} = \max f_j $
\begin{equation*}
\sum_{x \in V} \prod_j f_j(B_jx) = \sum_{x \in V}\prod_{p_j=1}f_j(B_jx)\prod_{p_k = 0}f_j(B_kx) \leq \prod_{p_j=1} \| f_j \|_1 \prod_{p_k = 0} \|f_k\|_\infty = \prod_{j} \|f_j\|_{1/p_j}
\end{equation*}
\end{proof}

We have given a Sum-of-Square Proof for (\ref{sos}), the degree of the proof is discussed by
\begin{lemma}
The degree of the Sum-of-Square proof for (\ref{sos}) is $O(n^mm^{m/2}+s\sum_{j=1}^ms_j)$
\end{lemma}
\begin{proof} To count the degree of the proof, we need to identify the polynomials with largest degree inside the proof. There are two cases for polynomials with largest degree.
\begin{enumerate}
\item The final inequality being proved has the largest degree
\item The inequality associated with extreme points has the largest degree
\end{enumerate}
For the first case, the degree the final inequality is $s\sum_{j=1}^ms_j$. For the second case, the degree of those inequalities is related to the fractional representation of $p$ as the extreme points of $Q(V)$. Since $Q(V)$ is a polytope defined by linear constraints, the extreme points of $Q(V)$ are basic feasible solutions of these constraints. By Cramer's rule,  the size of the fractional representation of $p$ is bounded by $\det(A)$ where $A$ is $m \times m$ coefficient matrix from the constraints. Since the each entry of $A$ is bounded by dimension $n$, by Hadamard's inequality \cite{antman1999jacques}, $\det(A) \leq n^mm^{m/2}$. So the degree of Sum-of-Square proof is poly$(n^mm^{m/2}, s\sum_{j=1}^ms_j)$. 
\end{proof}
So when the degree of the final inequality $s\sum_{j=1}^ms_j = O(1)$, the degree of the proof is $poly(d^{O(1)})$. This finish the proof of (\ref{sos}).


\bibliographystyle{plain}
\bibliography{Brascamp-Lieb}

\end{document}